\documentclass[journal]{IEEEtran}
\usepackage[utf8]{inputenc}
\usepackage[T1]{fontenc}
\usepackage{amsmath,amssymb,bm}
\usepackage{graphicx}
\usepackage{url}
\usepackage{color}
\usepackage{xcolor}

\usepackage{eso-pic}
\usepackage{amsmath,amssymb,amsthm,amsfonts,bm}
\usepackage{graphicx}
\usepackage{algorithm}
\usepackage[noend]{algpseudocode}
\usepackage{booktabs}
\usepackage{cite}
\usepackage{xcolor}
\usepackage{stmaryrd}
\usepackage{url}
\usepackage{tikz}
\usetikzlibrary{arrows.meta,positioning,decorations.pathreplacing,patterns,shapes.geometric}
\usepackage{url}
\usepackage[colorlinks=true,linkcolor=blue,citecolor=blue,urlcolor=blue]{hyperref}

\graphicspath{{figs/}}

\newtheorem{theorem}{Theorem}

\newtheorem{proposition}[theorem]{Proposition}
\newtheorem{remark}{Remark}

\newcommand{\vect}[1]{\bm{#1}}
\newcommand{\mat}[1]{\bm{#1}}
\newcommand{\set}[1]{\mathcal{#1}}
\newcommand{\expect}{\mathbb{E}}
\newcommand{\complex}{\mathbb{C}}

\begin{document}
\AddToShipoutPictureFG*{%
  \AtPageUpperLeft{%
    \raisebox{-0.30in}[0pt][0pt]{%
      \makebox[\paperwidth][c]{%
        \fontsize{7}{8}\selectfont
        \textit{This work has been submitted to the IEEE for possible
        publication. Copyright may be transferred without notice,
        after which this version may no longer be accessible.}%
      }%
    }%
  }%
}

\title{Visibility-Region Coupling in XL-MIMO AGV Fleets: Triple-Role Modeling and Masked Beamforming}

\author{
\IEEEauthorblockN{Changhao He,
                      Xiaojuan Zhang,~\IEEEmembership{Senior Member,~IEEE}}

 \thanks{C. H. He is with King Abdullah University of Science and Technology (KAUST), Saudi Arabia (changhao.he@kaust.edu.sa).
 X. J. Zhang is with the Institute of Advanced Intelligence and Computing, Agency for Science, Technology and Research (A*STAR IAIC), Singapore (xiaojuanzhang@ieee.org).
             Corresponding author: X. J. Zhang. }}

\maketitle

\begin{abstract}
Extremely large-scale multiple-input multiple-output (XL-MIMO) is a promising technology for supporting automated guided vehicle (AGV) fleets in smart port terminals. However, the metallic container environment induces spatial non-stationarity, whereby each AGV is visible to only a subset of the array, referred to as its visibility region (VR). Unlike existing XL-MIMO models that assume user-independent VRs, we show that each AGV simultaneously acts as a communication user, a metallic scatterer, and a blocker, resulting in coupled user channels and VRs. We formulate this \emph{triple-role} effect through a VR-coupled channel model and develop a VR-aware downlink beamforming framework based on masked weighted minimum mean-square error (WMMSE), where the masking operation exactly enforces VR support constraints while significantly reducing computational complexity. Simulation results in a realistic smart port scenario demonstrate more than a threefold sum-rate improvement over VR-unaware baselines, with the gains becoming increasingly pronounced as fleet density increases.
\end{abstract}



\section{Introduction}
\label{sec:intro}

\IEEEPARstart{M}{odern} automated container terminals rely on fleets of automated guided vehicles (AGVs) that require ultra-reliable and low-latency wireless connectivity for real-time coordination and navigation~\cite{Bauk2019}. Extremely large multiple-input multiple-output (XL-MIMO) systems, equipped with hundreds of antenna elements over large apertures, provide the array and multiplexing gains needed to support dense AGV deployments~\cite{Lu2024}. Owing to the large array aperture, however, users operate in the radiative near-field and experience spatially non-stationary channels, where each user is visible to only a subset of the array, referred to as its visibility region (VR)~\cite{Lu2024,Zhi2023}.

Container terminals further complicate wireless propagation. Dense metallic container stacks create severe blockage and strong specular reflections, resulting in highly location-dependent propagation characteristics~\cite{Willhammar2022}. Existing studies have investigated VR modeling~\cite{Zhi2023}, energy-efficient transmission under spatial non-stationarity~\cite{Zhang2024}, sub-array selection~\cite{Ubiali2026}, VR detection~\cite{Xu2024}, and modular XL-MIMO architectures~\cite{Mumtaz2025}. Nevertheless, these works generally assume user-independent propagation and treat VRs as static or independently varying across users.

A distinctive characteristic of AGV fleets is that every AGV simultaneously acts as a communication user, a metallic scatterer generating additional reflection paths, and a physical blocker that may obstruct propagation for neighboring AGVs~\cite{Qin2026}. Consequently, the channel and VR of one AGV depend not only on its own position but also on the locations of surrounding AGVs, giving rise to dynamic inter-user VR coupling that is absent from conventional XL-MIMO channel models. Exploiting this coupling is therefore essential for efficient resource allocation in dense AGV networks.

The main contributions of this letter are as follows.
\begin{itemize}
    \item We develop a near-field spatially non-stationary channel model that decomposes the NLoS channel into static container-induced and dynamic AGV-induced components. Based on this model, we characterize the \emph{AGV triple role} and establish a VR-blocking model that captures inter-user VR coupling.

    \item We propose a two-timescale VR-aware transmission framework. On the slow timescale, sub-array association is optimized via mixed-integer linear programming (MILP). On the fast timescale, downlink beamforming is solved using a VR-constrained weighted minimum mean-square error (WMMSE) algorithm that exploits the spatial sparsity introduced by VRs.

    \item Simulation results in a realistic container-terminal scenario demonstrate that the proposed scheme substantially outperforms conventional VR-unaware and static-VR baselines, achieving more than threefold improvement in sum rate and up to 35\% gain over a greedy association method in dense AGV deployments.
\end{itemize}

\section{System Model and Triple-Role Coupling}
\label{sec:system}

\begin{figure}[t]
    \centering
    \includegraphics[width=0.96\columnwidth]{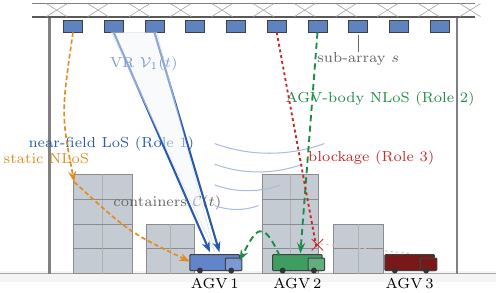}
    \caption{Smart-port XL-MIMO scenario. A gantry-mounted modular BS with $S$ sub-arrays serves AGVs in container canyons. AGV~2 is simultaneously a user (near-field LoS), a scatterer opening an NLoS path to AGV~1, and a blocker shadowing a sub-array link to AGV~3.}
    \label{fig:system}
\end{figure}



A modular XL-MIMO base station (BS) with $M=S\,M_s$ antennas, partitioned into $S$ sub-arrays of $M_s$ elements each, serves $K$ single-antenna AGVs on the downlink (Fig.~\ref{fig:system}). The BS is mounted at $h_{base} =20m$, while AGVs travel at $v=5-20km/h$. The container topology $\set{C}(t)$ is available from the terminal operating system and evolves much more slowly than AGV mobility. Unlike conventional XL-MIMO systems, each AGV simultaneously acts as a communication user, a metallic scatterer, and a blocker. For brevity, the slot index $t$ is omitted from instantaneous channel quantities.

\subsection{VR-Coupled Near-Field Channel}

Let $\delta_{k,s}\in\{0,1\}$ denote whether sub-array~$s$ is visible to AGV~$k$, with visibility region (VR) $\set{V}_k=\{s:\delta_{k,s}=1\}$. The aggregate channel is
\begin{equation}
    \vect{h}_k=\sum_{s=1}^{S}\delta_{k,s}\vect{h}_{k,s},
\end{equation}
where each visible sub-array follows the near-field Rician model
\begin{equation}\label{eq:rician}
    \vect{h}_{k,s}=\sqrt{\beta_{k,s}}\!\left(\sqrt{\tfrac{\kappa_{k,s}}{\kappa_{k,s}+1}}\,\vect{a}_s(\vect{p}_k)+\sqrt{\tfrac{1}{\kappa_{k,s}+1}}\,\vect{h}_{k,s}^{\text{NLoS}}\right)\!,
\end{equation}
with $\beta_{k,s}$ the path gain, $\kappa_{k,s}$ the Rician factor, and $\vect{h}_{k,s}^{\text{NLoS}}$ defined in~\eqref{eq:nlos}. The near-field steering vector is
\begin{equation}\label{eq:steering}
    [\vect{a}_s(\vect{p}_k)]_m=\exp\!\Big(\!-j\tfrac{2\pi}{\lambda}\big(\|\vect{p}_k-\vect{u}_{s,m}\|-\|\vect{p}_k-\vect{u}_{s,0}\|\big)\Big),
\end{equation}
where $\vect{p}_k$ is the AGV position, $\vect{u}_{s,m}$ is the $m$-th antenna element of sub-array~$s$, and $\vect{u}_{s,0}$ is its phase-reference center. Unlike conventional XL-MIMO models with static or user-independent VRs, $\set{V}_k$ depends jointly on the AGV position, the container topology, and neighboring AGVs, leading to coupled, time-varying visibility regions.

\subsection{Triple-Role VR Coupling}

The NLoS component comprises static container reflections and dynamic AGV-body scattering:
\begin{equation}\label{eq:nlos}
\begin{split}
    \vect{h}_{k,s}^{\text{NLoS}}={}&\underbrace{\textstyle\sum_{\ell=1}^{L_{k,s}^{\text{C}}} g_{k,s,\ell}^{\text{C}}\,\vect{a}_s(\vect{q}_{k,s,\ell}^{\text{C}})}_{\text{container reflections}}\\[-2pt]
    &+\underbrace{\textstyle\sum_{j\neq k}\sum_{\ell=1}^{L_{k,s,j}^{\text{A}}} g_{k,s,j,\ell}^{\text{A}}\,\vect{a}_s(\vect{q}_{k,s,j,\ell}^{\text{A}})}_{\text{AGV-body scattering}},
\end{split}
\end{equation}
where container reflections depend on the topology $\set{C}(t)$, whereas AGV-induced paths depend on neighboring positions $\{\vect{p}_j\}_{j\neq k}$ and disappear when $L^{\text{A}}_{k,s,j}=0$.

Each AGV simultaneously acts as a communication \emph{user}, a metallic \emph{scatterer}, and a physical \emph{blocker}. The blocking effect modifies the VR according to
\begin{equation}\label{eq:blocking}
    \delta_{k,s}=\delta_{k,s}^{\text{C}}\cdot\textstyle\prod_{j\neq k}\bigl(1-b_{k,s,j}\bigr),
\end{equation}
where $\delta_{k,s}^{\text{C}}$ is the container-only VR indicator and $b_{k,s,j}\in\{0,1\}$ indicates whether AGV~$j$ blocks the dominant path between sub-array~$s$ and AGV~$k$. Consequently,
\begin{equation}\label{eq:vr-coupled}
    \set{V}_k=f\bigl(\vect{p}_k,\,\set{C}(t),\,\{\vect{p}_j\}_{j\neq k}\bigr),
\end{equation}
so both $\vect{h}_k$ and $\set{V}_k$ depend on the positions of \emph{all} AGVs. Unlike conventional XL-MIMO models with independent user channels, the resulting inter-user coupling is dynamic yet predictable because AGV trajectories are coordinated by the terminal operating system.


\subsection{Received Signal}

With beamformer $\vect{w}_k\in\complex^M$ and unit-power symbol $x_k$, AGV~$k$ receives
\begin{equation}\label{eq:rx}
    y_k=\vect{h}_k^H\vect{w}_k x_k+\sum_{j\neq k}\vect{h}_k^H\vect{w}_j x_j+n_k,
\end{equation}
where $n_k\sim\set{CN}(0,\sigma^2)$. Since $\vect{h}_k$ is determined by the fleet-coupled VR model, both the desired signal and interference depend on the positions of neighboring AGVs. The resulting SINR and achievable rate are
\begin{equation}\label{eq:sinr}
    \gamma_k=\frac{|\vect{h}_k^H\vect{w}_k|^2}{\sum_{j\neq k}|\vect{h}_k^H\vect{w}_j|^2+\sigma^2},\qquad
    R_k=\log_2(1+\gamma_k),
\end{equation}
subject to the total transmit-power  $\sum_k\|\vect{w}_k\|^2\le P_{\text{tot}}$.

\section{Two-Timescale VR-Aware Beamforming}
\label{sec:solution}

We maximize the PF-weighted sum-rate over the slow-stage association $\mat{A}\in\{0,1\}^{S\times K}$ ($A_{s,k}=1$ if sub-array~$s$ may serve AGV~$k$) and fast-stage beamformer $\mat{W}$:
\begin{align}
    \max_{\mat{A},\mat{W}}\;&\textstyle\sum_{k=1}^{K}\alpha_k R_k \label{eq:opt}\\
    \text{s.t.}\;\; & A_{s,k}\le\delta_{k,s},\;\; \textstyle\sum_k A_{s,k}\le N_s^{\max},\;\; \textstyle\sum_s A_{s,k}\le L_k^{\max},\notag\\
    & \mathrm{supp}(\vect{w}_k)\subseteq\{s:A_{s,k}=1\},\;\; \textstyle\sum_k\|\vect{w}_k\|^2\le P_{\text{tot}},\notag
\end{align}
where $\alpha_k=1/\bar{R}_k$ is the PF weight ($\bar{R}_k$ the historical mean rate), $N_s^{\max}$ the per-sub-array user capacity, and $L_k^{\max}$ the per-user sub-array budget. Problem~\eqref{eq:opt} is mixed-integer and non-convex, and its variables live on different timescales: $\mat{A}$ need change only with $\set{C}(t)$, whereas $\mat{W}$ tracks fast fading and motion. We therefore decompose it.

\subsection{Slow Stage: Association}
Over a topology epoch we maximize a long-term surrogate. With $\bar{\delta}_{k,s}=\expect[\delta_{k,s}]$ and $\bar{\beta}_{k,s}=\expect[\beta_{k,s}]$ estimated by Monte-Carlo sampling of the fleet through the coupled map~\eqref{eq:vr-coupled}, we solve the binary linear program $\max_{\mat{A}}\sum_{s,k}U_{s,k}A_{s,k}$ subject to the first three constraints of~\eqref{eq:opt}, with utility
\begin{equation}\label{eq:utility}
    U_{s,k}=\bar{\beta}_{k,s}-\gamma\!\textstyle\sum_{j\neq k}\bar{\delta}_{j,s}.
\end{equation}
The first term favors strong, reliably visible links; the contention penalty ($\gamma>0$) discourages overloading popular sub-arrays. Crucially, $\bar{\delta}_{k,s}$ follows~\eqref{eq:blocking} over sampled positions, embedding AGV-blocking statistics; associating on the static $\bar{\delta}^{\text{C}}_{k,s}$ instead systematically over-commits shadowed sub-arrays.


\subsection{Fast Stage: Masked WMMSE}

Given $\mat{A}$ and instantaneous CSI, the \emph{effective} active set is
$\set{A}_k=\{s:A_{s,k}\delta_{k,s}=1\}$, i.e., the intersection of the association and instantaneous VR. The beamformer is restricted to $\set{I}_k=\bigcup_{s\in\set{A}_k}\set{M}_s$ through the diagonal mask $\mat{D}_k$, where $[\mat{D}_k]_{mm}=1$ iff $m\in\set{I}_k$. Thus, once $\mat{A}$ is fixed, the support is fixed and the beamforming subproblem becomes continuous. Applying WMMSE~\cite{Shi2011}, the MMSE receiver, error, and PF-weighted MSE weight are
\begin{equation}\label{eq:wmmse-ue}
    u_k=\frac{\vect{h}_k^H\vect{w}_k}{\sum_j|\vect{h}_k^H\vect{w}_j|^2+\sigma^2},\;\;
    e_k=1-u_k^*\vect{h}_k^H\vect{w}_k,\;\;
    \omega_k=\frac{\alpha_k}{e_k},
\end{equation}
and the masked beamformer update is
\begin{equation}\label{eq:wmmse-v}
    \vect{w}_k=\mat{D}_k\Bigl(\textstyle\sum_j\omega_j|u_j|^2\,\vect{h}_j\vect{h}_j^H+\mu\mat{I}\Bigr)^{-1}\!\omega_k u_k\,\vect{h}_k,
\end{equation}
where $\mu\ge0$ is obtained by bisection to satisfy the power constraint. Since $\mat{D}_k$ suppresses inactive antennas, \eqref{eq:wmmse-v} operates only on the active set $\set{I}_k$.

After each slot, the PF weights are updated as
\begin{equation}\label{eq:pf-update}
    \bar{R}_k\leftarrow(1-\rho)\bar{R}_k+\rho R_k,\qquad
    \alpha_k=1/\bar{R}_k,
\end{equation}
so AGVs experiencing persistent blockage are assigned higher priority, improving long-term fairness.

\subsection{Disjoint-VR Structure}

The following proposition characterizes the structural objective of the slow-stage association. Let $\set{A}_k$ denote the effective active sub-array set of AGV~$k$, and let $\set{I}_k=\bigcup_{s\in\set{A}_k}\set{M}_s$ be its corresponding antenna support.

\begin{proposition}\label{prop:disjoint}
Suppose the effective active sets are pairwise disjoint and mutually invisible, i.e., for all $k\neq j$, $\set{A}_k\cap\set{A}_j=\varnothing$ and $\delta_{k,s}=0$ for every $s\in\set{A}_j$. Then: \emph{(i)} the fast-stage beamforming problem decouples across users; \emph{(ii)} the multi-user interference vanishes, i.e., $\vect{h}_k^H\vect{w}_j=0$ for all $j\neq k$; and \emph{(iii)} the PF-weighted sum-rate $\sum_k\alpha_kR_k$ attains its interference-free upper bound.
\end{proposition}

\begin{IEEEproof}
Since $\vect{w}_j=\mat{D}_j\vect{w}_j$ is supported on $\set{I}_j=\bigcup_{s\in\set{A}_j}\set{M}_s$, while the mutual-invisibility condition implies $\mat{D}_j\vect{h}_k=\vect{0}$ for all $k\neq j$, it follows that
\[
\vect{h}_k^H\vect{w}_j
=\vect{h}_k^H\mat{D}_j\vect{w}_j
=(\mat{D}_j\vect{h}_k)^H\vect{w}_j
=0.
\]
Hence the SINR in~\eqref{eq:sinr} reduces to
$\gamma_k=\frac{|\vect{h}_k^H\vect{w}_k|^2}{\sigma^2},$ which depends only on $\vect{w}_k$. Therefore, the PF objective becomes separable across users, and each user is independently optimized by the maximum-ratio beamformer $\vect{w}_k\propto\mat{D}_k\vect{h}_k$ under the power constraint, establishing the proposition.
\end{IEEEproof}

\begin{remark}
Proposition~\ref{prop:disjoint} identifies mutually invisible, disjoint VRs as the ideal interference-free operating point. Although this condition is generally unattainable in dense AGV deployments due to residual visibility and scattering, it provides the design target for the slow-stage association. Specifically, the contention penalty in~\eqref{eq:utility} discourages assigning sub-arrays that are visible to multiple AGVs, thereby steering the association toward disjoint VRs whenever permitted by the fleet geometry. The remaining overlap is subsequently mitigated by the masked WMMSE beamforming.
\end{remark}


\subsection{Complexity}

The slow stage is dominated by Monte Carlo estimation of the coupled VR statistics, requiring $\set{O}(N_{\text{s}}KS(K+N_{\text{box}}))$ operations. Since it is executed only once per topology epoch, its amortized per-slot complexity is negligible.

For the fast stage, let $M_k=|\set{I}_k|\le L_k^{\max}M_s$ denote the number of active antennas for AGV~$k$. Each WMMSE iteration requires receiver/weight updates, construction of the active covariance sub-block, an $\set{O}(M_k^3)$ eigendecomposition, and power-control bisection, yielding
\begin{equation}\label{eq:complexity}
    \set{O}\!\big(T\textstyle\sum_k(M_k^3+KM_k^2+T_\mu M_k)\big)
    =\set{O}\!\big(TKM_{\text{a}}^3\big),
\end{equation}
where $M_{\text{a}}=\max_k M_k$. In contrast, conventional full-array WMMSE scales as $\set{O}(TKM^3)$ with $M=SM_s$. Since masking restricts each user to $M_{\text{a}}\le L_k^{\max}M_s\ll M$, the per-user complexity is reduced by approximately $(S/L_k^{\max})^3$.



\section{Simulation Results}
\label{sec:results}

We consider a $400\times300$\,m smart container terminal with eight storage lanes (stack height $12$\,m, lane width $15$\,m). Channels are generated using a geometry-based model: LoS is determined by 3D ray-box intersection, while NLoS consists of first- and second-order image-method reflections from PEC container walls. AGV bodies are modeled as $3\times1.5\times1.5$\,m PEC plates for both scattering and blocking. Unless otherwise stated, the system uses $f_c=3.5$\,GHz, $M=512$ antennas ($S=16$, $M_s=32$), $P_{\text{tot}}=40$\,dBm, $N_s^{\max}=2$, $L_k^{\max}=6$, $\gamma=0.02$, $T=5$ WMMSE iterations, and $\rho=0.2$. The container topology is updated every $30$\,min with $10\%$ of stacks reshuffled, and all results are averaged over five topology realizations and 100 time slots.

We compare the proposed scheme with three baselines: \textbf{B1}, a far-field stationary model with full-array WMMSE; \textbf{B2}, the proposed near-field model using only container-induced VRs (no AGV blocking or scattering); and \textbf{B3}, the full channel model with greedy sub-array association based on instantaneous $\beta_{k,s}$. Performance is evaluated in terms of average sum-rate, outage probability, and the CDF of the per-user rate.
\begin{figure}[ht]\centering
\includegraphics[width=0.70\columnwidth]{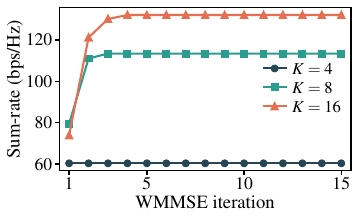}
\vspace{-10pt}
\caption{Convergence of the masked WMMSE for $K\in\{4,8,16\}$.}
\label{fig:convergence}
\end{figure}

\emph{Convergence (Fig.~\ref{fig:convergence}).} Despite restricting each user to its support $\set{I}_k$, the sum-rate increases monotonically and stabilizes within $3$--$4$ iterations for all fleet sizes, confirming that the projection mask $\mat{D}_k$ preserves the stationary-point convergence of standard WMMSE; we adopt $T=5$ as a conservative default. Combined with the $(S/L_k^{\max})^3$ per-user cost reduction of~\eqref{eq:complexity}, the fast stage is both convergent and real-time.

\begin{figure}[ht]\centering
\includegraphics[width=0.70\columnwidth]{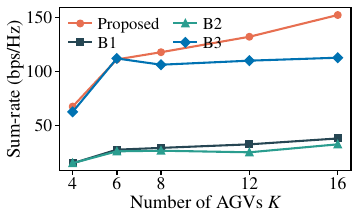}
\vspace{-10pt}
\caption{Mean sum-rate vs.\ number of AGVs $K$ (speed $10$\,km/h).}
\label{fig:sumrate}
\end{figure}
\emph{Sum-rate (Fig.~\ref{fig:sumrate}).} The proposed method beats VR-unaware B1/B2 by a wide margin for all $K$, confirming that ignoring near-field geometry and the triple role wastes power on blocked or poorly focused sub-arrays. Against greedy B3 the two are essentially tied at low density ($K=6$), but the MILP association pulls ahead as $K$ grows $+10.9\%$ at $K=8$, $+20.1\%$ at $K=12$, $+35.0\%$ at $K=16$. This widening gap is the signature of global load balancing under the tight capacity $N_s^{\max}=2$: greedy overloads popular sub-arrays and starves contenders, whereas the MILP spreads load, so its advantage grows precisely as the fleet densifies.

\begin{figure}[ht]\centering
\includegraphics[width=0.70\columnwidth]{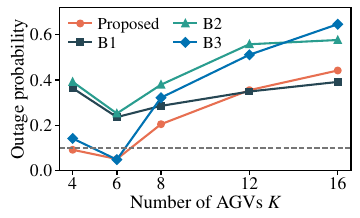}
\vspace{-10pt}
\caption{Outage probability ($R_k<R_{\min}$) vs.\ $K$ (speed $10$\,km/h).}
\label{fig:outage}
\end{figure}

\emph{Outage (Fig.~\ref{fig:outage}).} All methods show a shallow minimum near $K=6$; below it the occasional deep-canyon AGV inflates outage, above it contention drives it up. Among resource-constrained schemes the proposed method has the lowest outage: $0.44$ at $K=16$, vs.\ $0.58$ (B2) and $0.65$ (B3). Consistent with its rate distribution, it also serves the largest fraction of high-rate users and keeps the fewest in complete outage. The far-field B1 attains a slightly lower outage at high $K$ ($0.39$) only by spreading power across all $M$ antennas without the capacity limit $N_s^{\max}$: this egalitarian allocation barely clears $R_{\min}$ for most users while sacrificing over $4\times$ the sum-rate (Fig.~\ref{fig:sumrate}). Under the same resource budget the proposed design degrades the most gracefully while delivering far higher throughput. The absolute outage is high at large $K$ because $N_s^{\max}=2$ fundamentally limits how many AGVs can be served at once.

    \begin{figure}[ht]
        \centering
        \includegraphics[width=0.7\linewidth]{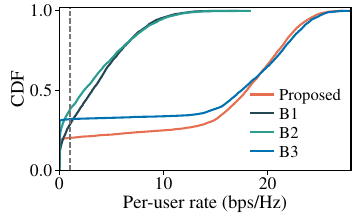}
        \vspace{-10pt}
        \caption{CDF of per-user rate at $K=8$, speed 10\,km/h.}
        \label{fig:cdf}
    \end{figure}
\emph{CDF (Fig.~\ref{fig:cdf}).} The VR-unaware baselines B1 and B2 are concentrated at low rates, with median rates of approximately $3$ and $2.3$\,bps/Hz, respectively, and nearly all users below $15$\,bps/Hz. This highlights the cost of ignoring near-field effects and the AGV triple role. In contrast, the proposed scheme shifts the distribution substantially to the right, achieving a median rate of $17.9$\,bps/Hz and a $90$th-percentile rate of $22.9$\,bps/Hz. It also reduces the outage probability ($<\!1$\,bps/Hz) to $0.20$, compared with $0.32$ for the greedy baseline B3. While some users remain in the moderate-rate region ($1$--$15$\,bps/Hz) due to the limited sub-array capacity ($N_s^{\max}=2$), the proposed scheme achieves the largest fraction of high-rate users and the lowest outage probability.
\section{Conclusion}
We identified the AGV's triple role in smart-port XL-MIMO, where each AGV acts as a user, scatterer, and blocker, coupling user VRs and channels. We proposed a two-timescale masked WMMSE design that enforces VR support and reduces per-user complexity by $(S/L_k^{\max})^3$. We also established a disjoint-VR condition under which beamforming decouples and the PF rate bound is attained. Simulations show over $3\times$ the sum-rate of VR-unaware baselines, with gains growing with fleet density. Predicting this coupling from AGV trajectories for proactive VR management is left for future work.

\bibliographystyle{IEEEtran}
\bibliography{ref}

\end{document}